\documentclass[a4paper, 12pt]{article}

\usepackage[a4paper,left=3cm,right=3cm,top=2.5cm,bottom=3.0cm]{geometry}

\usepackage{graphicx}
\usepackage{graphics}
\usepackage{caption}
\usepackage{subcaption}
\usepackage{color}
\usepackage{amsmath}
\usepackage{amssymb}
\usepackage{amsfonts}
\usepackage[english]{babel}
\usepackage{enumitem}
\usepackage{comment}
\usepackage{gensymb}
\usepackage{mathtools}
\usepackage{bbm}
\usepackage[nocompress]{cite}

\usepackage{amsthm}

\newtheorem{theorem}{Theorem}[section]
\newtheorem{lemma}{Lemma}[section]
\newtheorem{observation}{Observation}[section]
\newtheorem{remark}{Remark}[section]

\newcommand{\nat}{\mathbb{N}}
\newcommand{\natplus}{\mathbb{N}^+}
\newcommand{\ratplus}{\mathbb{R}^+}
\newcommand{\LambdaSet}{\mathbf{\Lambda}}

\newcommand{\leaf}[1]{\mathcal{L}(#1)}

\newcommand{\subtree}[1]{T_#1}

\newcommand{\cT}{\mathcal{T}}

\newcommand{\cG}{\mathcal{G}}

\newcommand{\cS}{\mathcal{S}}
\newcommand{\cC}{\mathcal{C}}

\newcommand{\RingOff}{\texttt{\textup{RingOffline}}}
\newcommand{\RingOn}{\texttt{\textup{RingOnline}}}

\newcommand{\ProcLabel}{\texttt{\textup{SetLabeling}}}
\newcommand{\ProcStrat}{\texttt{\textup{SetStrategy}}}
\newcommand{\ProcMain}{\texttt{\textup{CostExpl}}}

\usepackage{algpseudocode}
\usepackage[nothing]{algorithm}
\floatname{algorithm}{Procedure}
\algtext*{EndWhile}
\algtext*{EndIf}
\algtext*{EndFor}

\usepackage{epstopdf}

\ifpdf
\fi
\begin{document}

\renewcommand{\thefootnote}{\fnsymbol{footnote}}

\footnotetext[2]{Faculty of Electronics, Telecommunications and Informatics, Gda{\'n}sk University of Technology, Gda{\'n}sk, Poland}

\title{Minimizing the Cost of Team Exploration\thanks{Research partially supported by National Science Centre (Poland) grant number 2015/17/B/ST6/01887.}}
\author{Dorota Osula\footnotemark[2]}

\maketitle
\thispagestyle{empty}

\floatname{algorithm}{Procedure}

\begin{abstract}
A group of \emph{mobile  agents} is given a task to \emph{explore} an edge-weighted graph $G$, i.e., every vertex of $G$ has to be visited by at least one agent. There is no centralized unit to coordinate their actions, but they can freely communicate with each other. The goal is to construct a deterministic \emph{strategy} which allows agents to complete their task optimally.
In this paper we are interested in a \emph{cost-optimal} strategy, where the cost is understood as the total distance traversed by agents coupled with the cost of invoking them. Two graph classes are analyzed, rings and trees, in the \emph{off-line} and \emph{on-line} setting, i.e., when a structure of a graph is known and not known to agents in advance. We present algorithms that compute the optimal solutions for a given ring and tree of order $n$, in $O(n)$ time units. For rings in the on-line setting, we give the $2$-competitive algorithm and prove the lower bound of $3/2$ for the competitive ratio for any on-line strategy. For every strategy for trees in the on-line setting, we prove the competitive ratio to be no less than $2$, which can be achieved by the $DFS$ algorithm. 
\end{abstract}

\vspace*{1cm}

{\bf keywords:} graph exploration, distributed searching, cost minimization, mobile agents, on-line searching.

\section{Introduction}

A group of \emph{mobile agents} is given a task to \emph{explore} the edge-weighted graph $G$, i.e., every vertex of $G$ has to be visited by at least one agent. Initially agents are placed on one vertex, called \emph{homebase}\footnote{After finishing the exploration, agents do not have to come back to the homebase.}, they are distinguishable (each entity has its unique \emph{id}) and they can communicate freely during the whole exploration process. The goal is to find a deterministic \emph{strategy} (\emph{protocol} or \emph{algorithm}), which is a sequence of \emph{steps}, where each step consists of parallel \emph{moves}. Each move is one of the two following types: (1) traversing an edge by an agent or (2) invoking a new agent in the homebase. The strategy should be optimal in specified sense; in the literature we discuss the following approaches: \emph{exploration time}, \emph{number of entities}, \emph{energy} and \emph{total distance} optimization. Exploration time is the number of \emph{time units} required to complete the exploration, with the assumption that a walk along an edge $e$ takes $w(e)$ time units (where $w(e)$ is the weight of the edge $e$). As one agent is sufficient to explore the whole graph, in the problem of minimizing the number of entities additional restrictions of the size of the \emph{battery} of searchers (i.e., the maximum distance each agent can travel) or the maximum exploration time are added. Energy is understood as the maximum value taken over all agents traversed distances. Lastly, the total distance is the sum of distances traversed by all agents. In this work we introduce a new approach: we are looking for the \emph{cost-optimal} strategy, where cost is the sum of the distances traversed by agents and a cost of invoking them. We consider the problem in the \emph{off-line} setting, where a graph is known in advance for searchers and the \emph{on-line} one, where agents have no \emph{a priori} knowledge about the graph.
We assume, for simplicity, that in one step only one agent can perform a move.\footnote{One may notice, that in order to reduce the number of time units of the algorithm, agents moves, when possible, should be perform simultaneously.} As the measure for an on-line algorithm the \emph{competitive ratio} is used (formally defined later), which is the maximum taken over all networks of the results of the on-line strategy divided by the optimal one in the off-line setting.

\noindent
\textbf{Related Work} For exploration in the off-line setting (referred often as \emph{searching}) many different models were extensively studied and numerous deep results have been obtained. Interestingly, there is a strong connection between graph exploration and many different graph parameters, e.g., pathwidth, tree\-width, vertex separation number; see e.g., \cite{FT08} for a survey and further references. In \cite{czyzowicz2017energy} edge-weighted trees in the off-line setting were studied, where a group of $k$ mobile agents has a goal to explore the  tree minimizing the total distance. Agents (as in the model presented in this paper) do not have to return to the homebase after the exploration. Thus, for $k$ big enough, it is a special case of our model, for which the invoking cost is equal to zero. For $k$ greater or equal to the number of leaves authors present the $O(n)$ time algorithm solving the problem. In the on-line setting, the most results were established in minimizing the time of the exploration. Algorithms and bounds of competitive ratio were investigated, mostly for trees \cite{FraigniaudCollectiveTree,DyniaWhyRobots,HigashikawaExploration} in different communication settings.  
As for the edge-exploration of general graphs (where apart from vertices also every edge has to be explored) see \cite{FraigniaudCollectiveTree,brass2014improved}.
The competitive ratio of exploration arbitrary graphs for teams of size bigger than $\sqrt{n}$ was studied in \cite{disser2016general,DereniowskiFastCollaborative}.
As for different graph classes, grids \cite{OrtolfMultiRobotExplo} and rings \cite{HigashikawaExploration} were investigated.
Several studies have been also undertaken minimizing the energy \cite{DyniaPowerAware,DyniaWhyRobots} and
the number of agents \cite{DasCollaborativeExploration} for trees. As one can observe trees are very important graph class and in this paper we give the algorithms in off-line and on-line setting and prove the upper bound for the competitive ratio. We note that our results may be of particular interest not only by providing theoretical insight into searching dynamics in agent computations, but may also find applications in the field of robotics. This model describes well the real life problems, where every traveled unit costs (e.g., used fuel or energy) and entities costs itself (e.g., equipping new machines or software license cost). It can be viewed also as a special case of traveling salesmen, vehicle routing or pickup and delivery problem. Many deep 
results were established in these fields, see e.g., \cite{kumar2012survey,golden2008vehicle,berbeglia2007static,vaishnav2017traveling,bellmore1968traveling} for the further references.

This work is constructed as follows: in the next Section we introduce the necessary notation and formally define the problem. The further two Sections present results for rings. In Section \ref{sec:RingsOff} the cost-optimal algorithm for the off-line setting is presented, whereas in Section \ref{sec:RingsOn} the $2$-competitive algorithm in the on-line setting is described. It is also proved, that for a positive invoking cost and any on-line strategy there exist a ring, which force the strategy to produce at least $3/2$ times higher cost than the optimal, off-line one.
Section \ref{sec:TreesOff} contains the algorithm and its analysis for trees in the off-line setting, while Section \ref{sec:TreesOn} provides the proof that no algorithm can perform better on trees in the on-line setting than $DFS$. In other words, the competitive ratio for every on-line algorithm is no less than $2$. We finish this work with the summary and a future outlook, and suggest areas of further research. 

\section{Notation}\label{sec:notation}

Let $\cG$ be a class of undirected, edge-weighted graphs. For every $G = (V(G),E(G), w) \in \cG$ we denote the sets of vertices and edges of $G$ as $V = V(G)$ and $E = E(G)$, respectively, $n = |V|$ and an edge-weight function $w : E \rightarrow \ratplus$. The sum of all weights of a subgraph $H$ of $G$ is denoted by $w(H) = \sum_{e \in E(H)} w(e)$. For any graph $G \in \cG$ we denote as $W(u,v)$ a walk 
that starts in $u \in V$ and finishes in $v \in V$ (if $u = v$ then either walk is a single vertex or a closed walk). A path is understood as an open walk with no repeated vertices. The \emph{distance} between two vertices is the sum of weights of all edges in a shortest path connecting them. For any graph $G$ and $v,u \in V(G)$ we denote the distance between $v$ and $u$ by $d_G(v, u)$ and omit the bottom index, when $G$ is clear from the context.

For every tree $T$ and the pair of vertices $v,u \in V(T)$ we denote a path between them as $P_T(v,u)$ and omit the bottom index, when a graph is clear from the context. 
 For any tree $T$ and vertex $v \in V(T)$ we define \emph{branch} as a subtree rooted in a child $c$ of $v$ enlarged by the vertex $v$ and edge $(v,c)$.

We define a strategy $\mathcal{S}$ as a sequence of moves of the following two types: (1) traversing an edge by an agent and (2) invoking a new agent in the homebase. We say that a strategy \emph{explores} a vertex, when it is reached for the first time. Let $k \in \natplus$ be the number of agents used by $\mathcal{S}$ (notice that $k$ is not fixed) and $d_i \in \ratplus \cup \{0\}$ the distance that $i$-th agent traversed, $i=1,2,\ldots,k$. The \emph{invoking} cost $q \in \ratplus \cup \{0\}$ is the cost connected to the agents: every time the strategy uses a new agent it has to `pay' for it $q$.
In other words, before exploring any vertex, the strategy needs to decide what is more profitable: invoke a new agent (and pay for it $q$) or use an agent already present in the graph. The number of agents, that can be invoked, is unbounded. The cost $c$ is understood as the sum of invoking costs and the total distance traversed by entities, i.e., $c = kq + \sum_{i=1}^kd_i$. The goal is to find a cost-optimal strategy, which explores every graph $G\in\cG$.

Let $\cS$ be an on-line strategy and $\cS^{opt}$ be the cost-optimal, off-line strategy for every graph in $\cG$. We denote as $\cS(G)$ and $\cS^{opt}(G)$ the cost of proceeding the strategy $\cS$ and $\cS^{opt}$, respectively, on $G \in \cG$.
As a measure for an on-line algorithm $\cS$ the \emph{competitive ratio} is used, which is the maximum taken over all networks of the results of the on-line strategy divided by the optimal one in the off-line setting, i.e.
\begin{equation*}
r(\cS) = \max_{G \in \cG}\frac{\cS(G)}{\cS^{opt}(G)}.
\end{equation*}

In the on-line setting it is assumed that an agent, which occupies the vertex $v$, knows the length of edges incident to $v$ and the \emph{status} of vertices adjacent to $v$, i.e., if they have been already explored. We assume that agents can freely communicate with each other.

\section{Rings in the Off-line Setting}\label{sec:RingsOff}

Let $\cC \subset \cG$ be a class of undirected, edge-weighted rings. For every  $C = (V, E, w) \in \cC$ of order $n$, we denote the vertices as $V = \{v_i,\ i\in \{0,\ldots,n - 1\}\}$ and edges as $E = \{e_i = (v_i,v_{i+1}),\ i\in \{0,\ldots,n-2\}\} \cup e_{n-1} = (v_{n-1},v_0)$. Without loss of generality, let  a homebase of $C$ be in $v_0$. We define the problem in the off-line setting as follows:

\begin{description}
\item[Off-line Ring Problem Statement]$\ $\\[2mm]
Given the ring $C$, the invoking cost $q$ and the homebase $h$, find a strategy of the minimum cost.
\end{description}

In the cost-optimal solution exactly one of the edges does not have to be traversed. Procedure $\RingOff$ finds in $O(n)$ steps, which edge is optimal to omit. If this edge is incident to the homebase, then only one agent is used, which simply traverses the whole ring without it (lines \ref{alg:first1}-\ref{alg:first2} and \ref{alg:last1}-\ref{alg:last2}). Otherwise, depending on the cost $q$, there might be one or two agents in use. 
Let $e$ be an omitted edge and let $C' = C \backslash e$, i.e., $C'$ is a tree rooted in $v_0$ with two leaves. We denote as $v^{min}$ and $v^{max}$ the closer and further, respectively, leaf in $C'$. If the invoking cost $q$ is lower than $d_{C'}(v_0,v^{min})$, then it is more efficient to invoke two agents, which traverse two paths $P_{C'}(v_0,v^{min})$ and $P_{C'}(v_0,v^{max})$ (lines \ref{alg:middle21}-\ref{alg:middle22}). On the other hand, if $q \geq d_{C'}(v_0,v^{min})$, then only one agent is used, which traverses the path $P_{C'}(v_0,v^{min})$ twice (lines \ref{alg:middle11}-\ref{alg:middle12}). 

We give a formal statement of the procedure $\RingOff$ and make an observation about its cost-optimality.

\begin{algorithm}
   \caption{$\RingOff$}
   \begin{algorithmic}[1]
   	\Require Ring $C$, homebase $v_0$, invoking cost $q$
    \Ensure Strategy $\cS$    
    \State $C_i \leftarrow C \backslash e_i,\ i \in \{0,\ldots,n-1\}$
    \State $v^{min}_i \leftarrow$ a vertex $v \in \{v_i, v_{i+1}\}$ for which $d_{C_i}(v_0,v)$ is minimum, $i\in\{1,\ldots,n-2\}$
    \State $v^{max}_i \leftarrow$ a vertex $v \in \{v_i, v_{i+1}\}$ for which $d_{C_i}(v_0,v)$ is maximum, $i\in\{1,\ldots,n-2\}$
    \State $c_i \leftarrow q + w(C_i),\ i=0,n-1$
    \State $c_i \leftarrow  \min\{2q + w(C_i), q + d_{C_i}(v_0,v^{min}_i) + w(C_i)\},\ i\in \{1,\ldots,n-2\}$
        \State Let $i_{min}$ be the index of the minimum element of $\{c_i,\ i\in\{0,\ldots,n-1\}\}$
        \State Add a move to $\cS$: invoke an agent $a_1$ in $v_0$
        \If{$i_{min} ==0$} \label{alg:first1}   		
       		\State Add a sequence of moves to $\cS$: traverse by $a_1$ path $P_{C_{i_{min}}}(v_0,v_1)$ \label{alg:first2}
        \ElsIf {$i_{min}>0$ and $i_{min}<n-1$}
        	\If{$2q + w(C_i) < q + d_{C_i}(v_0,v^{min}_{i_{min}}) + w(C_i)$}\label{alg:middle21}
            	\State Add a sequence of moves to $\cS$: traverse by $a_1$ path $P_{C_{i_{min}}}(v_0,v^{min}_{i_{min}})$
            	\State Add a move to $\cS$: invoke an agent $a_2$ in $v_0$
            	\State Add a sequence of moves to $\cS$: traverse by $a_2$ path $P_{C_{i_{min}}}(v_0,v^{max}_{i_{min}})$ \label{alg:middle22}          
            \Else \label{alg:middle11}
        		\State Add a sequence of moves to $\cS$: traverse by $a_1$ path $P_{C_{i_{min}}}(v_0,v^{min}_{i_{min}})$ 
                \State Add a sequence of moves to $\cS$: traverse by $a_1$ path $P_{C_{i_{min}}}(v^{min}_{i_{min}},v^{max}_{i_{min}})$ \label{alg:middle12}
           \EndIf
        \Else \label{alg:last1}
        	\State Add a sequence of moves to $\cS$: traverse by $a_1$ path $P_{C_{i_{min}}}(v_0,v_{n-1})$  \label{alg:last2}
        \EndIf
        \State return $\cS$
  \end{algorithmic}
\end{algorithm}

\begin{observation}
For any invoking cost $q$ and ring $C$, the strategy $\cS$, returned by the procedure $\RingOff$ is cost-optimal.
\end{observation}

\begin{proof}
Let $C$ be any ring, $q$ any invoking cost and $\cS$ the strategy returned by the procedure $\RingOff$ for $C$ and $q$. Because every vertex has to be explored, exactly one edge does not have to be traversed. Our procedure computes for every edge $e \in E(C)$, the optimal cost $c$ of exploring $C' = C \backslash e$. If $e$ is incident to the homebase, then $C'$ is a path with a homebase in one of its ends. Thus, the cost-optimal strategy uses one agent and $c = q + w(C')$. If $e$ is not incident to the homebase, then $C'$ is a path with a homebase in one of its internal vertices. Thus, the cost-optimal strategy uses one or two agents. In the first case, an agent has to traverse to the closer end of $C'$ and then along the whole path. In the second case, each of the agents traverses from the homebase to one of the end vertices. Thus, $c = \min\{q + d' + w(C'), 2q + w(C')\}$, where $d'$ is the distance between the homebase and the closer end vertex in $C'$.
At the end $\RingOff$ chooses an edge, which deleting leads to the lowest cost and sets the corresponding strategy.
\end{proof}

\section{Rings in the On-line Setting}\label{sec:RingsOn}

In this section we present the procedure $\RingOn$, which produces in $O(n)$ steps an $2$-competitive strategy, which explores any unknown ring $C$. We also prove the lower bound of $3/2$ for the competitive ratio for any $q>0$. We define the problem in the on-line setting as follows:

\begin{description}
\item[On-line Ring Problem Statement]$\ $\\[2mm]
Given the invoking cost $q$ and the homebase $h$, find a strategy of the minimum cost for any ring $C$. 
\end{description}

We start by giving the informal description of the procedure $\RingOn$. Let $\cS$ be the on-line strategy returned by the procedure $\RingOn$ for a given homebase $v_0$ and invoking cost $q$. Firstly $\cS$ invokes an agent $a_1$ in $v_0$ and denotes as $e_1$ and $e_{-1}$ edges incident to $v_0$, with the lower and higher weight respectively (lines \ref{alg:invoke}-\ref{alg:firstEdges}). Searcher $a_1$ traverses first $e_1$ and then continues the exploration process as long as it is profitable, i.e., the cost of traversing the next edge is less or equal to the invoking cost plus $w(e_{-1})$ (lines \ref{alg:ring1a}-\ref{alg:ring1b}). If at some point a new agent is invoked, then it traverses the edge $e_{-1}$ (lines \ref{alg:ring2a}-\ref{alg:ring2b}). We notice here that the lines \ref{alg:ring2a}-\ref{alg:ring2b} are executed at most once, as these are initial steps for the second agent. Later, the greedy approach is performed: an edge with lesser weight is traversed either by $a_1$ (lines \ref{alg:ring1a}-\ref{alg:ring1b}) or by $a_2$ (lines \ref{alg:ring3a}-\ref{alg:ring3b}). Below we give a formal statement of the procedure $\RingOn$.

\begin{algorithm}
   \caption{$\RingOn$}
   \begin{algorithmic}[1]
   	\Require Homebase $v_0$, invoking cost $q$
    \Ensure Strategy $\cS$      	
        \State $i_r \leftarrow 1$
        \State $i_l \leftarrow -1$
        \State $s \leftarrow 1$
        \State Add a move to $\cS$: invoke an agent $a_1$ in $v_0$  \label{alg:invoke}
        \State Denote as $e_1$ and $e_{-1}$ edges adjacent to $v_0$, with the lower and higher weight respectively \label{alg:firstEdges}
        \While{Graph is not explored}
        	\While{$(w(e_{i_l}) + q\cdot s )\geq w(e_{i_r})$ and graph is not explored}\label{alg:ring1a} 
            	\State Add a move to $\cS$: traverse $e_{i_r}$ by $a_1$
                \State Denote the unexplored edge incident to the vertex occupied by $a_1$ as $e_{i_r + 1}$
                \State $i_r \leftarrow i_r + 1$
            \EndWhile\label{alg:ring1b}
            \If{$s==1$ and $(w(e_{-1}) + q) < w(e_{i_r})$} \label{alg:ring2a}
            		\State Add a move to $\cS$: invoke an agent $a_2$ in $v_0$ 
                	\State Add a move to $\cS$: traverse $e_{-1}$ by $a_2$ 
                    \State Denote the unexplored edge incident to the vertex occupied by $a_2$ as $e_{-2}$
                \State $i_l \leftarrow -2$
                \State $s \leftarrow 0$
                \EndIf\label{alg:ring2b}
            \If{$s==0$} \label{alg:ring3a}
            \While{$w(e_{i_l}) < w(e_{i_r})$ and graph is not explored} 
            	\State Add move to $\cS$: traverse $e_{i_l}$ by $a_2$ 
                \State Denote the unexplored edge incident to the vertex occupied by $a_2$ as $e_{i_l - 1}$
                \State $i_l \leftarrow i_l - 1$
            \EndWhile
            
            \EndIf \label{alg:ring3b}
        \EndWhile
        \State return $\cS$
  \end{algorithmic}
\end{algorithm}

The next lemma says that for any invoking cost and any ring procedure $\RingOn$ returns the solution at most twice worse than the optimum, which is tight. 

\begin{lemma}
The strategy returned by $\RingOn$ is $2$-competitive.
\end{lemma}

\begin{proof}
Let $C \in \cC$ be any ring for which the cost-optimal, off-line strategy $\cS^{opt}$ uses two agents or omits an edge incident to the homebase. We notice that procedure $\RingOn$ computes the cost-optimal strategy for $C$. However, the situation is different otherwise.

Let now $C \in \cC$ be any ring with the homebase in $v_0$ for which the cost-optimal strategy uses one agent and omits the edge not incident to the homebase. Let $q$ be any invoking cost and $\cS$ be a strategy returned by the procedure $\RingOn$. Denote as $e_{max}$ the edge of $C$ of the maximum weight and as $e'$ the edge incident to $v_0$ of the bigger weight. Let $e$ be an omitted edge in the cost-optimal off-line strategy and $v^{min}$ be the closer leaf in the tree $C\backslash e$ rooted in $v_0$. The cost-optimal strategy $\cS^{opt}$ uses one agent, that traverses firstly to the $v^{min}$, then returns to the homebase and explores the rest of the ring apart from the edge $e$. Thus, its cost can be lower bounded by:
\begin{equation}
\cS^{opt}(C) = q + w(C \backslash e) + d_{C\backslash e}(v_0,v^{min}) > q + w(C \backslash e).
\end{equation}

If $q + w(e') < w(e_{max})$, then
$\cS$ uses two searchers and omits $e_{max}$, i.e.,
\begin{align}
\cS(C) = 2q + w(C \backslash e_{max}).
\end{align}

On the other hand, if $q + w(e') \geq w(e_{max})$, then
$\cS$ invokes one searcher, which traverses the whole ring apart from $e'$, i.e,
\begin{align}
\cS(C) = q + w(C \backslash e') = q + w(C\backslash e_{max}) + w(e_{max}) - w(e') \leq 2q + w(C \backslash e_{max}). 
\end{align}

This leads to the following upper bound for the competitive ratio

\begin{align}
r(\cS) &\leq \frac{2q + w(C \backslash e_{max})}{q + w(C \backslash e)} \leq \frac{2q + w(C \backslash e)}{q + w(C \backslash e)} \leq 2.
\end{align}

The bound of $2$ can be reached for the three vertices ring $C' \in \cC$, where $w(e_0) = w(e_2) = 1$ and $w(e_1) = q$ for $q$ large enough. Indeed, we notice that although both strategies, $\cS$ and $\cS^{opt}$, invoke only one searcher, in the cost-optimal solution it traverses the edge $e_0$ twice instead of traversing the edge $e_1$, i.e., $\cS^{opt}(C') = q + 3$ and $\cS(C') = 2q + 1$, which finishes the proof.

\end{proof}
   
The theorem below shows that for any positive invoking cost and any on-line strategy there exist  a ring for which the strategy achieves at least $3/2$ times higher cost than the optimal one.

\begin{theorem}
For any invoking cost $q>0$, every on-line strategy $\cS$ is at least $\frac{3}{2}$-competitive.
\end{theorem}

\begin{proof}
Let $\epsilon$ be a small positive number, such that $\epsilon \ll q$ and $q\mod \epsilon=0$. Let $\cC_1,\ \cC_2 \subset \cC$ be two classes of rings constructed as follows. For every $i \in \{3,4,\ldots\}$ we add to the class $\cC_1$ a ring $C_1^i$ 
of order $i$ with a homebase in $v_0$ and set the weights of all edges as $\epsilon$. For every positive integer $i$ and $j \in \{i+2,i+3,\ldots\}$ we add to the class $\cC_2$ a ring $C_2^{(i,j)}$ 
of order $j$ with homebase in $v_0$. We set the weights of all edges, apart from $(v_i,v_{i+1})$ with the weight $2q$, as $\epsilon$. Let $\cC' = \cC_1 \cup \cC_2$.

We are going to show that for any on-line strategy $\cS$, there exist a ring $C \in \cC'$ and a strategy $\cS'$, such that $\cS(C) \geq \frac{3}{2}\cS'(C)$.

Let $G_i, i \in \nat$, be an explored part of the graph in the $i$-th step of the algorithm, starting from $G_0 = v_0$.
Firstly, any strategy $\cS$ invokes in $v_0$ one agent $a_1$ and explores some part of a ring. Let $j_1 \in \nat$ be a step in which the second agent $a_2$ is invoked by $\cS$ (we set $j_1$ as infinity, if $\cS$ uses only one searcher). For every $i < j_1$, if during the $i$-th step $a_1$ explores a new edge $(v,v')$ of the weight $\epsilon$, a new vertex $u$ and edge $(v',u)$ are added. The weight of $(v',u)$ is set as $\epsilon$, if $w(P_{G_i}(v_0,v')) < q$, and $2q$, if $w(P_{G_i}(v_0,v')) = q$. We consider two cases. In the first one, when the new agent is invoked only edges of the length $\epsilon$ are visible. In the second case, $a_1$ reaches a vertex incident to the edge of the weight $2q$ before the second agent is invoked or $a_1$ explores the graph on its own.

\paragraph*{Case A: \textmd{\textit{in the step $j_1$ only edges of the length $\epsilon$ are visible.}}} We can treat $G_{j_1}$ as a tree rooted in  $v_0$, where $v_0$ has two branches. Denote the number of vertices in them as $1 \leq h_1 < q$ and $0 \leq h_2 \leq h_1$ (where by $h_2 = 0$, we understand that $G_{j_1}$ is a path of the length $\epsilon \cdot h_1$ starting in $v_0$). We omit the case when $a_2$ is invoked in the second step (i.e., $h_1 = h_2 = 0$), as then competitive ratio of at least $2$ can be easily obtained (e.g., for a triangle with edges of the weight $\epsilon$). 
We choose a ring $C = C_1^{h_1 + h_2 + 2}$ from the class $\cC_1$. See Figure \ref{fig:onlineC1} for the illustration. 

We notice that $|V(C)| = h_1 + h_2 + 2$ and $|V(G_{j_1})| = h_1 + h_2 + 1$. 
In order to explore $G_{j_1}$, $a_1$ traversed at least twice the path $P_{G_{j_1}}(v_0,v_{h_1 + 2})$ (possible empty if $h_2=0$) and once $P_{G_{j_1}}(v_0,v_{h_1})$. Then, in the $j_1$-th step the second agent is invoked, which generates the extra cost $q$. In order to explore the whole $C$, at least one extra move of cost $\epsilon$ has to be done (e.g., $a_1$ can traverse from $v_{h_1}$ to $v_{h_1 + 1}$). 
Thus, the total cost of exploring $C$ by $\cS$ can be lower bounded by
\begin{align}
\cS(C) \geq 2q + \epsilon(2h_2 + h_1 + 1).
\end{align}

Let $\cS'$ be a strategy, which explores $C$  by visiting all vertices exactly once with one agent. This leads to the following upper bound of the cost-optimal solution
\begin{equation}
\cS^{opt}(C) \leq \cS'(C) =  q + \epsilon(h_1 + h_2 + 1).
\end{equation}

As $\epsilon h_1 < q$ and $h_2 \geq 0$, we obtain the following lower bound of the competitive ratio
\begin{align}
r(\cS)& = \lim\limits_{\epsilon \rightarrow 0} \frac{\cS(C)}{\cS^{opt}(C)} \geq 
\lim\limits_{\epsilon \rightarrow 0} \frac{2q + \epsilon(h_1 + 2h_2 + 1)}{q + \epsilon(h_1 + h_2 + 1)} \geq \lim\limits_{\epsilon \rightarrow 0} \frac{2q + q + 2\cdot 0 + \epsilon}{q + q + 0 + \epsilon} = \\
& = \lim\limits_{\epsilon \rightarrow 0}\frac{3q + \epsilon}{2q + \epsilon} = \frac{3}{2}.
\end{align}

\paragraph*{Case B: \textmd{\textit{$a_1$ reaches a vertex incident to the edge of the weight $2q$ before the second agent is invoked or $a_1$ explores the graph on its own.}}}  Let $j_2$ be the step in which $a_1$ explores the vertex incident to the edge of the weight $2q$. We can treat $G_{j_2}$ as a tree rooted in  $v_0$, where $v_0$ has two branches. Denote the number of vertices in them as $h_1 = q/\epsilon$ and $0 \leq h_2 < h_1$. We choose a ring $C = C_2^{(h_1 + h_2 + 2, h_1)}$ from the class $\cC_2$. See Figure \ref{fig:onlineC2} for the illustration. Let $\cS'$ be a strategy, which uses one agent, which firstly explores vertices $v_0, v_{h_1+h_2+1}, v_{h_1 + h_2},\ldots,v_{h_1+1}$, then returns to $v_0$ and explores the rest of $C$ omitting the edge of the weight $2q$. This leads to the following upper bound of the cost-optimal solution 
\begin{equation}
\cS^{opt}(C) \leq \cS'(C) =  q + 2\epsilon(h_2 + 1) + \epsilon h_1 = 2q + 2\epsilon(h_2 + 1).
\end{equation}

In the $j_2$-th step only one vertex of $C$ is not explored by $\cS$: $v_{h_1+1}$, which is incident to the edges of the weights $\epsilon$ and $2q$, and agent $a_1$ occupies the vertex $v_{h_1}$. The remaining vertex can be either explored by $a_1$ or by a newly invoked searcher $a_2$.

\paragraph*{Subcase B1: \textmd{\textit{$a_1$ explores the vertex $v_{h_1+1}$.}}} In the cheapest solution searcher $a_1$ returns to $v_0$ and traverses the path $v_0, v_{h_1+h_2+1}, v_{h_1 + h_2},\ldots,v_{h_1+1}$ of the length $\epsilon(h_2 + 1)$. Thus, the total cost of exploring $C$ by $\cS$ can be lower bounded by
\begin{align}
\cS(C) \geq q + \epsilon(2h_2 + 2h_1 + h_2 + 1)= 3q + \epsilon(3h_2 + 1),
\end{align}

which leads to the following competitive ratio
\begin{equation}
r(\cS) = \lim\limits_{\epsilon \rightarrow 0} \frac{\cS(C)}{\cS^{opt}(C)} \geq \lim\limits_{\epsilon \rightarrow 0} \frac{3q + 3\epsilon h_2 + \epsilon}{2q + 2\epsilon h_2 + 2\epsilon} = \frac{3}{2}.
\end{equation}

\paragraph*{Subcase B2: \textmd{\textit{a newly invoked agent $a_2$ explores the vertex $v_{h_1+1}$.}}} The total cost of exploring $C$ by $\cS$ can be lower bounded by
\begin{align}
\cS(C) \geq (q + 2\epsilon h_2 + \epsilon h_1) + (q + \epsilon h_2 + 1) = 3q + 3\epsilon h_2 + \epsilon,
\end{align}

giving the same bound of $\frac{3}{2}$ and finishing the proof.


\begin{figure}[htb]
  \centering
  \begin{subfigure}{.45\textwidth}
  \centering
  \includegraphics[scale=1.1]{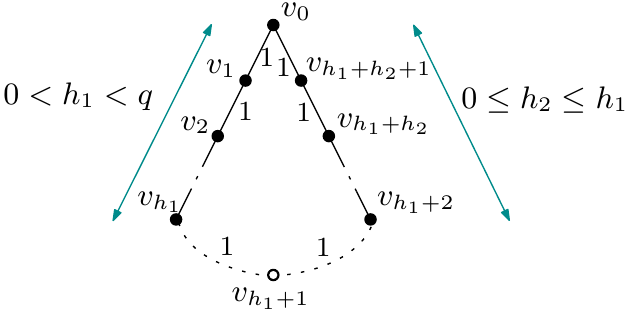}
  \caption{Ring $C_1^{h_1 + h_2 + 2}\in \cC_1$.}
  \label{fig:onlineC1}
\end{subfigure} \quad
\begin{subfigure}{.45\textwidth}
  \centering
  \includegraphics[scale=1.1]{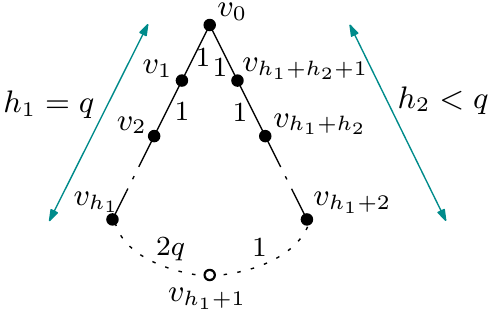}
  \caption{Ring $C_2^{(h_1 + h_2 + 2, h_1)}\in \cC_2$.}
  \label{fig:onlineC2}
\end{subfigure}
\caption{Illustration of rings from the classes $\cC_1$ and $\cC_2$; \emph{black dotes} and \emph{solid lines} denote already explored part of the graph, whereas \emph{circles} and \emph{dashed lines} stand for the unvisited part of the rings.}
\label{fig:onlineC}
\end{figure}

\end{proof}

At the end we observe, that for $q=0$ the strategy returned by the procedure $\RingOn$ is cost-optimal for every ring.

\section{Trees in the Off-line Setting}\label{sec:TreesOff}

Let $T = (V,E,w) \in \cG$ be a tree rooted in a homebase $r$ and $\leaf{T}$ be the set of all leaves in $T$. For every $v \in V$  we denote by $\subtree{v}$ a subtree of $T$ rooted in $v$, $c(v)$ list of its children and $p(v)$ its parent vertex.

Vertex $v \in V$ is called a \emph{decision vertex} if $|c(v)| \geq 2$ and an \emph{internal vertex} if $|c(v)| = 1$ and $v$ is different from the root. We say that an agent \emph{terminates} in $v \in V$, if $v$ is its last visited vertex. 
We state the problem in the off-line setting formally:

\begin{description}
\item[Off-line Tree Problem Statement]$\ $\\[2mm]
Given the tree $T$, the invoking cost $q$ and the homebase in the root of $T$, find a strategy of the minimum cost.
\end{description}

\subsection{The Algorithm}\label{sec:algorithm}

In order to simplify our algorithm, a \emph{compressing} operation on a tree $T$ is proceeded. Let $v \in V(T)$ be a decision vertex and $u \in V(T)$ be a decision vertex, a leaf or the root. 
The new tree $T'$ is obtained by substituting every path $P_T(v,u)$, which apart from $u$ and $v$ consists only internal vertices, with a single edge $e=(v,u)$. The weight of $e$ is set as the weight of the whole path, i.e., 
$w(e) = w(P_T(v,u))$. 
See Figure \ref{fig:comp1} for an example of the compressing operation.


\begin{figure}[htb]
  \centering
  \begin{subfigure}{.45\textwidth}
    \centering
  \includegraphics[scale=0.7]{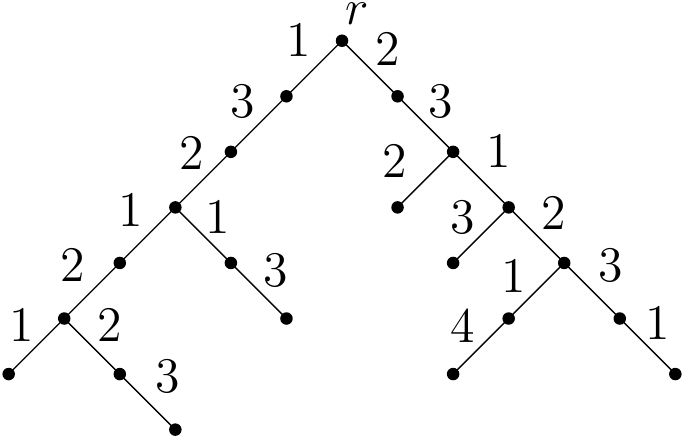}
  \caption{The original tree $T$.}
  \end{subfigure} \quad
  \begin{subfigure}{.45\textwidth}
    \centering
  \includegraphics[scale=0.9]{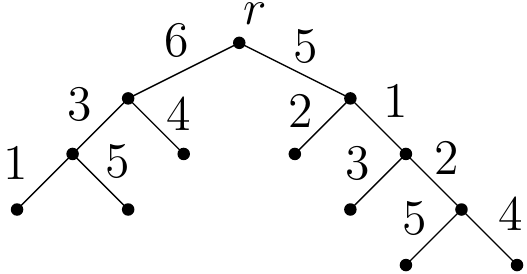}
  \caption{Tree $T$ after the compressing operation.}
  \end{subfigure}
\caption{The compressing operation on an exemplary tree $T$. The new tree $T'$ has no internal vertices.}
\label{fig:comp1}
\end{figure} 

\begin{observation}\label{obs:leaf}
\textit{In every cost-optimal strategy if an agent enters a subtree $T_v$, it has to explore at least one leaf in it.}
\end{observation}

\begin{proof}
Let $T$ be any tree, $v \in V$ and $a$ be an agent, which at some point
occupies $v$. If $a$ returns to $p(v)$ or terminates before exploring at least one leaf in $T_v$, then its moves inside $T_v$ can be omitted.
Indeed, in this situation $a$ has reached vertices either (1) already explored, or (2) one lying on the path between $v$ and an unexplored leaf, which will be visited later anyway (by $a$ or any other agent).
\end{proof}

\begin{observation}\label{obs:not_returning}
\textit{In every cost-optimal strategy once an agent leaves any subtree, it never comes back to it.}
\end{observation}
\begin{proof}
Let $T$ be any tree rooted in $r$ and $v\in V$ different than $r$. By contradiction, let $\cS$ be the cost-optimal strategy for $T$, in which an agent $a$ after leaving $T_v$ returns to it. Denote as $l_a$ the leaf in which $a$ terminates (not necessarily $l_a \in \leaf{T_v}$). We split the walk $W(r,l_a)$ traversed by $a$ into parts. We define:
\begin{itemize}
\item $W^1(r,p(v))$ as a walk that $a$ traverses until it reaches $v$, excluding~$v$;
\item $W^2(v,v)$ as a walk that $a$ traverses inside $T_v$ before it leaves it;
\item $W^3(p(v),v)$ as a walk that $a$ traverses after leaving $T_v$ and before reaching~$v$;
\item $W^4(v,l_a)$ as a walk that $a$ traverses after $W^3(p(v),v)$.
\end{itemize}

In other words
\begin{align}
    W(r,l_a) & = W^1(r,p(v)) \circ (p(v),v) \circ W^2(v,v) \circ \\
    & \circ (v,p(v)) \circ W^3(p(v),v) \circ W^4(v,l_a). 
\end{align}

The total weight is then the following sum:
\begin{align}
    w(W(r,l_a)) & = w(W^1(r,p(v))) + w((p(v),v)) + w(W^2(v,v)) + \\
    & + w((v,p(v)))+ W^3(p(v),v) + w(W^4(v,l_a)). 
\end{align}

Let $S'$ be a strategy in which $a$ traverses: (1) $W^1(r,p(v))$, (2) $W^3(p(v),v)$, (3) $W^2(v,v)$ and (4) $W^4(v,l_a)$. Then 
\begin{equation}
  w(W(r,l_a)) = w(W^1(r,p(v))) + W^3(p(v),v) + w(W^2(v,v)) + w(W^4(v,l_a))  
\end{equation}

and $\cS'(T) < \cS(T)$, which is a contradiction that $\cS$ is cost-optimal.
\end{proof}

\begin{remark}\label{rem:compress}
Let $v$ be any internal vertex.
It is never optimal for an agent, which occupies $v$, to return to the previously occupied vertex in its next move.
\end{remark}
\begin{proof}
 Let $T$ be any tree and $v \in V$ be any internal vertex. 
Assume that at some step of the optimal strategy, agent $a$ occupies $v$. 
If the last traversed edge of $a$ is $(p(v),v)$, then it follows directly from Observation~\ref{obs:leaf}. On the other hand, the remark is true otherwise from Observation~\ref{obs:not_returning}.
\end{proof}

In other words, it is always optimal for agents to continue movement along the path once entered. 
Thus, if we find the optimal strategy for compressed tree $T'$, then we can easily obtain the optimal strategy for $T$.
The only difference is that instead of walking along one edge $(v,u)$ in $T'$, the agent has to traverse the whole path $P_T(v,u)$ in $T$. From now on, till the end of this Section, whenever we talk about trees, we refer to its compressed version.

For all vertices $v \in V(T)$ we consider a labeling $\Lambda_v$, which is a triple $(k,u_l,u_c)$, where $k$  stands for the minimum number of agents needed to explore the whole subtree $T_v$ by any cost-optimal strategy. The second one, $u_l$, is the furthest leaf from $v$ in $T_v$ (if there is more than one, then $v$ is chosen arbitrary) and $u_c$ is the child of $v$, such that $u_l \in T_{u_c}$. We will refer to this values using the dot notation, e.g., the number of agents needed to explore tree rooted in $v$ is denoted by $\Lambda_v.k$. The set of labels for all vertices is denoted by $\LambdaSet = \{\Lambda_v,\ v \in V(T)\}$.

\subsubsection{Procedures}

The algorithm is built on the principle of dynamic programming: first the strategy is set for leaves, then gradually for all subtrees and finally for the root. We present three procedures: firstly, labeling $\LambdaSet$ is calculated by $\ProcLabel$, which is the main core of our algorithm. Once labels for all the vertices are set, the procedure $\ProcStrat$ builds a strategy based on them. The main procedure $\ProcMain$ describes the whole algorithm.

\paragraph*{Procedure $\ProcLabel$} for every subtree $T_v$, calculates and returns labeling 
$\Lambda_v$. We give a formal statement of the procedure and its informal description followed by an example. Firstly, for every leaf $v$ label $\Lambda_v = (1,v,\text{null})$ is set, as one agent is sufficient to explore $v$. Then, by recursion, labels for the ancestors are set until the root $r$ is reached. Let us describe now how the labeling for the vertex $v$ is established based on the labeling of its children (main loop, lines \ref{alg:loopS}-\ref{alg:loopE}).  Firstly, the number of needed agents for $v$ is increased by the number of needed agents for its child $u$ (line \ref{alg:lab1a}). Then, if the distance between $v$ and the furthest leaf in $T_u$ (i.e., $d(v,\Lambda_u.u_l)$) is less or equal to the distance from the root $r$ to $v$ plus the invoking cost $q$, the number of required agents is reduced by $1$ (lines \ref{alg:lab1b}-\ref{alg:lab1c}). Intuitively, it is more efficient to reuse this agent, than to invoke a new one from $r$. As we show formally later at most one agent can be returned, and it can happen only if $\Lambda_u.k = 1$.
Meanwhile the child of $v$, which is an ancestor of the furthest leaf in $T_v$ is being set (lines \ref{alg:lab1d}-\ref{alg:lab1e}). See the formal statement of the procedure and an example on the Figure~\ref{fig:setlabel1}. 

\begin{algorithm}
   \caption{$\ProcLabel$}
   \begin{algorithmic}[1]
   	\Require Tree $T$, vertex $v$, invoking cost $q$, labeling $\LambdaSet$
    \Ensure Updated $\LambdaSet$
    	\If{$v \in \leaf{T}$}
            \State $\Lambda_v \leftarrow (1,v,\text{null})$
            \State return $\LambdaSet$
        \EndIf       
    \ForAll{$u \in c(v) $}
            \State Invoke Procedure $\ProcLabel$ for $T, u, q$ and $\LambdaSet$       
      \EndFor
   
		\State $k, d^{max} \leftarrow 0$
        \State $u_c^{max} \leftarrow \text{null}$       
        \State $d_r \leftarrow d(r, v) + q$
    	\ForAll{$u \in c(v)$}  \label{alg:loopS}
            \State $k \leftarrow k + \Lambda_u.k$ \label{alg:lab1a}
            \State $d \leftarrow d(v, \Lambda_u.u_l)$
        	\If{$\Lambda_u.k == 1$ and $d \leq d_r$} \label{alg:lab1b}
                  \State $k \leftarrow k - 1$                    
            \EndIf  \label{alg:lab1c} 
            \If{$d > d^{max}$} \label{alg:lab1d} 
                 \State $d^{max} \leftarrow d$
                 \State $u_c^{max} \leftarrow u$
            \EndIf \label{alg:lab1e} \label{alg:loopE}
        \EndFor   
        \State $k \leftarrow \max\{1,k\}$
       \State $\Lambda_v \leftarrow (k, \Lambda_{u_c^{max}}.u_l, u_c^{max})$
       \State return $\LambdaSet$
  \end{algorithmic}
  \end{algorithm}


\begin{figure}[htb]
  \centering
  \begin{subfigure}{.45\textwidth}
    \includegraphics[width=\textwidth]{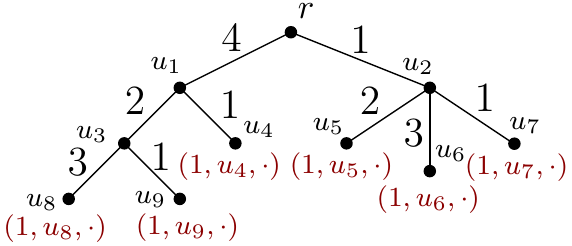}
    \caption{Firstly, labels for leaves are set.}
  \end{subfigure} \quad
  \begin{subfigure}{.45\textwidth}
    \includegraphics[width=\textwidth]{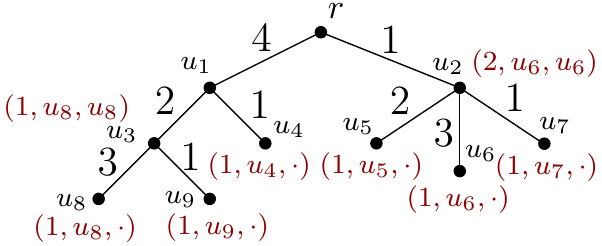}
    \caption{Then, gradually labels are being set for the ancestors, until the root is reached.}
  \end{subfigure} \\
   \begin{subfigure}{.45\textwidth}
    \includegraphics[width=\textwidth]{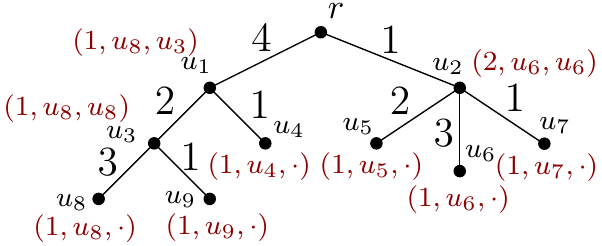}
    \caption{}
    \end{subfigure}\quad
    \begin{subfigure}{.45\textwidth}
    \includegraphics[width=\textwidth]{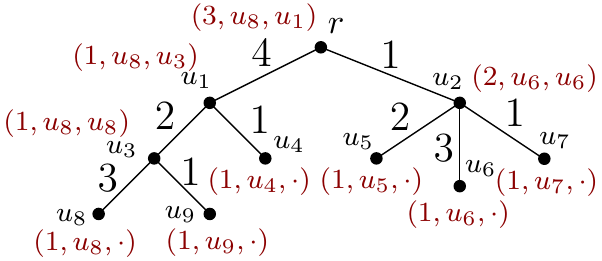}
    \caption{Three agents are required to explore this tree in the cost-optimal way.}
  \end{subfigure}
  
\caption{Example of the performing of the procedure $\ProcLabel$ for $q=0$. }
\label{fig:setlabel1}
\end{figure}

\paragraph*{Procedure $\ProcStrat$} builds a strategy for a given subtree $T_v$ based on the labeling $\LambdaSet$. If $v \in V(T)\backslash \leaf{T}$, then for each of its child $u$, firstly, the required number of agents is sent to $u$ (line \ref{alg:str1a}) and then the strategy is set for $u$ (line \ref{alg:str1b}). Lastly, for all children $u$ of $v$ (apart from the one, which has to be visited as the last one) if it is efficient for the searcher, which finished exploration of $T_u$ in $\Lambda_u.u_l$, to come back to $v$, then the `return' sequence of moves is added (lines \ref{alg:str1c}-\ref{alg:str1d}). It is crucial that for every $v$ the subtree $T_{\Lambda_v.u_c}$ is explored as the last one, but the order of the remaining subtrees is not important (line \ref{alg:children}).
To summarize, we give a formal statement of the procedure.

\begin{algorithm}
   \caption{$\ProcStrat$}
   \begin{algorithmic}[1]
   	\Require Tree $T$, vertex $v$, invoking cost $q$, labeling $\LambdaSet$, strategy $\mathcal{S}$
    \Ensure Strategy $\mathcal{S}$	
    	\If{$v \not\in \leaf{T}$}
        	\If{$v == r$}
            	\State Add a move to $\mathcal{S}$: invoke $\Lambda_r.k$ agents in $r$
            \EndIf
        \State $d_r \leftarrow d(r, v) + q$
        \State Let $c_1,\ldots,c_l$ be children of $v$, where $c_l = \Lambda_v.u_c$ \label{alg:children}
        \For{$i\in\{1,\ldots,l\}$}
        	\State Add a sequence of moves to $\mathcal{S}$: traverse $(v,c_i)$ by $\Lambda_{c_i}.k$ agents \label{alg:str1a}
            \State Invoke Procedure $\ProcStrat$ for $T,c_i,q,\LambdaSet$ and $\mathcal{S}$ \label{alg:str1b}
        	\If{$d(v, \Lambda_{c_i}.u_l) \leq d_r$ and $c_i \neq \Lambda_v.u_c$} \label{alg:str1c}
                  \State Add a sequence of moves to $\mathcal{S}$: send an agent back from $\Lambda_{c_i}.u_l$ to $v$                   
            \EndIf  \label{alg:str1d}    
           \EndFor
        \EndIf
  \end{algorithmic}
\end{algorithm}

\paragraph*{Procedure $\ProcMain$} consists of two procedures presented in the previous subsections. Firstly, $\ProcLabel$ is being invoked for the whole tree $T$. And then the strategy $\mathcal{S}$ is being calculated from the labeling $\LambdaSet$ by the procedure $\ProcStrat$. We observe that $\ProcMain$ finds a strategy in $O(n)$ time. To summarize, we give a formal statement of the procedure.

\begin{algorithm}
   \caption{$\ProcMain$}
   \begin{algorithmic}[l]
   	\Require Tree $T$, invoking cost $q$
    \Ensure Strategy $\mathcal{S}$
        \State Invoke Procedure $\ProcLabel$ for $T, r, q$ and $\emptyset$; set $\LambdaSet$ as an output
        \State Invoke Procedure $\ProcStrat$ for $T,r,q,\LambdaSet$ and $\emptyset$; set $\mathcal{S}$ as an output
        \State  Return $\mathcal{S}$
  \end{algorithmic}
\end{algorithm}

\subsection{Analysis of the Algorithm} \label{sec:results} 

In this Section, we analyze the algorithm by providing the necessary observations and lemmas
and give the lower and upper bounds. Firstly, let us make a simple observation about the behavior of agents in the cost-optimal strategies. 

\begin{observation}\label{obs:leaf_terminates}
\textit{In every cost-optimal strategy all agents terminates in leaves and every leaf is visited exactly once.}
\end{observation}

\begin{proof}
Let $T$ be any tree and $\cS$ be the cost-optimal strategy for $T$ in which an agent $a$ terminates in $v \in V(T)\backslash \leaf{T}$. By Observation~\ref{obs:leaf} agent $a$ has explored at least one leaf. Let $u \in \leaf{T}$ be the last explored leaf by $a$. 
Notice, that every vertex visited by $a$ after it leaves $u$ lies on a path between root and some other leaf, which means that either it has been already explored or will be later by some other agents. Thus, these moves are unnecessary and $\cS(T)$ is not minimal. The latter part of the observation is obvious.
\end{proof}

In our strategies, subtrees $T_v$ of the maximum height $d(r,v) + q$ are always explored by one agent. The next observation says that in the cost-optimal solution this agent finishes in the furthest leaf of $T_v$.

\begin{observation}\label{obs:furthest_leaf}
\textit{If one agent is cost-optimal to search a tree $T$, then it terminates in one of the furthest leaves.}
\end{observation}
\begin{proof}
Let $T$ be any tree rooted in $r$.
Let $\cS$ be the cost-optimal strategy for $T$, in which an agent $a$ terminates in leaf $l$ (Observation~\ref{obs:leaf_terminates}). 
Thanks to Observation~\ref{obs:not_returning} we notice that agent $a$ simply performs $DFS$ on $T$ truncated by moves from $l$ to $r$. In other words, the total cost is equal to $2w(T) - d(r,l)$, thus $l$ should be the furthest leaf.
\end{proof}

Let $v \in V(T)$ different then root. 
Lemma~\ref{lem:one_leaving} guarantees us that after the exploration of $T_v$ at most one agents returns to $p(v)$. Lemma~\ref{lem:label_ok} and Theorem~\ref{thm:main_offline} present our main results. 

\begin{lemma}\label{lem:one_leaving}
In every cost-optimal strategy if an agent leaves any subtree, it has explored it on its own.
\end{lemma}
\begin{proof}
Let $T$ be any tree rooted in $r$ and $v \in V$ different then $r$.
By contradiction, let $\cS$ be the cost-optimal strategy for $T$, in which $T_v$ is explored by at least two agents and at least one of them leaves $T_v$ at some step. 

Let $A$ be a group of agents, which terminates in leaves of $T_v$ and $B$ a group of agents, which visits at least one vertex of $T_v$, but terminates in leaves outside $T_v$. From the assumption we have that $|A \cup B| \geq 2$ and $B \neq \emptyset$. For every $a \in A \cup B$ let $u_a$ be the last visited leaf from $T_v$ (which existence is guaranteed by Observation~\ref{obs:leaf}) and let $a$ terminates in $l_a$. Thanks to the Observation~\ref{obs:not_returning} we can split the walk $W_a(r,l_a)$ traversed by every agent $a \in A \cup B$ into parts. We define:
\begin{itemize}
\item $W^1_a(r,p(v)),\ a \in A \cup B$ as a walk that $a$ traverses until it reaches $v$ for the first time, excluding~$v$;
\item $W^2_a(v,u_a),\ a \in A \cup B$ as a walk that $a$ traverses inside $T_v$ until it explores $u_a$;
\item $W^3_a(p(v),l_a),\ a \in B$ as a walk that $a$ traverses after leaving $T_v$, excluding $v$.
\end{itemize}

We obtain that, 
\begin{align}
   w(W_a(r,l_a)) & = w(W_a(r,u_a)) = w(W^1_a(r,p(v))) \\
   & + w((p(v),v)) + w(W^2_a(v,u_a)),
\end{align}
   for every $a \in A$, and 
   
   \begin{align}
       w(W_a(r,l_a)) & = w(W^1_a(r,p(v)) + w((p(v),v)) + w(W^2_a(v,u_a)) \\
       & + d(u_a,p(v)) +  w(W^3_a(p(v),l_a),
   \end{align}
  for every $a \in B$. We consider two cases.

\paragraph*{Case A: \textmd{\textit{$A \neq \emptyset$.}}}
We choose and arbitrary agent $a' \in A$ and modify its walk, so after $W^1_{a'}(r,p(v))$ it traverses all the walks $W^2_a(v,u_a),\ a \in B$ returning every time to $v$.  
All of the agents $a \in B$ traverses first $W^1_a(r,p(v))$ and then $W^3_a(p(v),l_a)$, i.e., there is no agent that leaves $T_v$. Obtained in that way $\cS'$ is a proper strategy, which explores the whole tree $T$. Let now $L$ and $L'$ be the total distances traversed by agents from $A \cup B$ in $\cS$ and $\cS'$, respectively. We get the following:

\begin{align}
L& = \sum_{a\in A}\left(w(W_a^1(r,p(v))) + w((p(v),v)) + w(W_a^2(v,u_a))\right) + \\
&+ \sum_{a\in B}(w(W_a^1(r,p(v))) + w((p(v),v)) + w(W_a^2(v,u_a)) + \\
& + d(u_a,p(v)) + w(W_a^3(p(v),l_a)
),\\
L'& = \sum_{a\in A}\left(w(W_a^1(r,p(v))) + w((p(v),v)) + w(W_a^2(v,u_a))\right) + \\ 
&+ \sum_{a\in B}\left(w(W_a^1(r,p(v)) + w(W_a^3(p(v),l_a))\right) + \\
&+ \sum_{a\in B}\left(w(W_a^2(v,u_a)) + d(u_a,v)\right)< L,
\end{align}

which finishes the proof of first case. 

\paragraph*{Case B: \textmd{\textit{$A = \emptyset$.}}} We choose and arbitrary agent $a' \in B$ and modify its walk, so after $W^1_{a'}(r,p(v))$ it traverses all the walks $W^2_a(v,u_a),\ a \in B$ returning every time to $v$. 
All of the other agents $a \in B$ traverses firstly $W^1_a(r,p(v))$ and then $W^3_a(p(v),l_a)$, i.e., only $a'$ leaves $T_v$. Obtained in that way $\cS'$ is a proper strategy, which explores the whole tree $T$. Similarly to the previous case one can show that $S'(T) < S(T)$, i.e., we get the contradiction that $S$ is cost-optimal.

\end{proof}

\begin{lemma}\label{lem:label_ok}
Let $\LambdaSet$ be a labeling returned by the procedure $\ProcLabel$ for an arbitrary tree $T$. Every cost-optimal strategy uses at least $\Lambda_v.k$ agents to explore 
$T_v,\ v\in V(T)$. \footnote{There exist cost-optimal strategies that can use more than $\Lambda_v.k$ agents. Indeed, if $d(v,\Lambda_v.u_l) = d(r,v) + q$ reusing the agent and calling a new one generates equal cost.} 
\end{lemma}
\begin{proof}
Let $T$ be any tree of the height $H$ rooted in $r$.
By the induction on the height of a tree $h$. Firstly, let $h = H$, i.e., we consider labeling of the set $\leaf{T}$, which is the base case in our procedure. Any cost-optimal strategy uses one agent to explore a leaf, thus $\Lambda_u.k = 1$ is correct. 

We assume now, that all the labeling is correct for vertices at levels greater than $h$ and consider vertices on the level $h$. Notice, that invoking $\ProcLabel$ in recursive way guarantees, that before the label on any vertex is computed, all of its children's labels are set. Let $v$ be any vertex of $T$ on level $0 \leq h < H$. The algorithm sets  
\begin{equation}
    \Lambda_v.k = \max\left\{\sum_{u \in c(v)}\left(\Lambda_u.k - \mathbbm{1}_{\left(d(v,\Lambda_u.u_l) \leq d(r,v)+q\right)} \right),1\right\}.
\end{equation}
Because once an agent leaves any subtree, it never comes back to it (Observation~\ref{obs:not_returning}), subtrees rooted in the children of $v$ can be searched sequentially.
Every $T_u,\ u \in c(v)$ needs $\Lambda_u.k$ agents, which is minimal from the induction assumption. We notice, that after exploring $T_u$ at most one agent might return to $v$ and be used to explore $T_v \backslash T_u$ (Lemma~\ref{lem:one_leaving}). It can happen only if $\Lambda_u.k = 1$. Thus, this agent has to finish in the furthest leaf of $T_u$ (Observation~\ref{obs:furthest_leaf}). The strategy $\cS$ reuses all agents (apart from the possible one from the last visited subtree $T_u$) which finishes the exploration of $T_u$ in the leaf, which is not `too far', i.e., $d(v,\Lambda_u.u_l) \leq d(r,v) + q$. Indeed, no other agent can be reused, because if $d(v,\Lambda_u.u_l) > d(r,v) + q$, then it is cheaper to call a new agent from the root. Thus, $\Lambda_v.k$ is minimal.
\end{proof}

\begin{theorem}\label{thm:main_offline}
Procedure $\ProcMain$ for every tree $T$ returns a strategy, which explores $T$ in the cost-optimal way.
\end{theorem}
\begin{proof}
Let $T$ be any tree rooted in $r$, $\LambdaSet$ be the labeling computed by $\ProcLabel$ and $\cS$ be a strategy for $T$ returned by $\ProcMain$. 
Constructing $\cS$ for every $T_v,\ v \in V(T)$ based on $\LambdaSet$ is straightforward. Let $u \in c(v)\backslash \{\Lambda_v.u_c\}$.
Firstly, traverse $\Lambda_u.k$ agents along $(v,u)$. Then, set the strategy for $T_u$.  After the exploration of $T_u$, if $d(v,\Lambda_u.u_l) \leq d(r,v) + q$, return an agent from $\Lambda_u.u_l$ to $v$. Repeat for every $u \in c(v)\backslash \{\Lambda_v.u_c\}$ in the random order. For the last child $u = \Lambda_v.u_c$ after the exploration of $T_u$ do not return any agents.

By Lemma~\ref{lem:label_ok} we know that $\Lambda_v.k$ is the minimum number of agents that has to be send to $v$ (or invoked, for $v = r$). The remaining thing to prove is that the order of exploring subtrees does not matter as long as the one with the furthest leaf is visited as the last one. Let 
\begin{align}
    U_1 & = \{u | u \in c(v),\ d(v,\Lambda_u.u_l) \leq d(r,v) + q\}, \\
    U_2 & = c(v)\backslash U_1, \\
    \cT_i & = \{T_{u} | u \in U_i\},\ i\in\{1,2\}, \\
    u_c &= \Lambda_v.u_c, \\
    T_c &= T_{u_c}.
\end{align}

 If $U_1 = \emptyset$, then $\Lambda_v.k = \sum_{u \in c(v)}\Lambda_u.k$ and the order of exploration does not matter. On the other hand, if $U_2 = \emptyset$, then $\Lambda_v.k = 1$ and the agent has to finish in the furthest leaf (Observation~\ref{obs:furthest_leaf}), i.e., $T_c$ has to be explored as the last one. Let then $U_1,U_2 \neq \emptyset$, we notice that 
$T_c \in \cT_2$. Our intuition may say that is better to first explore trees from $\cT_1$ and then from $\cT_2$. But as long as the last tree to visit is from $\cT_2$, the order of the rest subtrees does not influence $\Lambda_v.k$ or the cost. Let us consider the strategy, which firstly explores all trees from $\cT_1$. The total number of required agents is 
\begin{equation}
\Lambda_v.k = 1 + \left(\sum_{u \in U_2\backslash\{u_c\}}\Lambda_u.k -1\right) + \Lambda_{u_c}.k = \sum\limits_{u \in U_2}\Lambda_u.k.
\end{equation}

This amount does not change for any order of the exploration. Indeed, after the exploration of a tree from $\cT_1$ an agent is reused to explore the next tree, decreasing the same $\Lambda_v.k$ by one. See, for example, the strategy in which first are explored trees from $\cT_2 \backslash \{T_c\}$. 
 The total number of agents stays the same, i.e, 
 \begin{equation}
 \Lambda_v.k = \sum_{u \in U_2}\Lambda_u.k + 1 + (\Lambda_{u_c}.k - 1) = \sum_{u \in U_2}\Lambda_u.k.
 \end{equation}
 
 At the end we observe, that as trees are explored separately and $\Lambda_v.k$ is the same for any order of the exploration, the total cost is not influenced either.
\end{proof}

\noindent
\textbf{Lower and Upper Bounds} For any tree $T$ the value of the optimal cost $c$ is bounded by $q + w(T) \leq c \leq q + 2w(T) - H$. 
A trivial lower bound is achieved on the path graph, where one agent traverses the total distance of $w(T)$. The upper bound can be obtained by performing $DFS$ algorithm by one entity, which set it on $q + 2w(T)$. Let $DFS'$ be the modified version of $DFS$, such that the agent does not return to the homebase (i.e., terminates in one of the leaves). Then we get an improved upper bound of $q + 2w(T) - H$, where $H$ is the height of $T$, which is tight (e.g., for paths). 
It is worth to mention that although $DFS'$ performs well on some graphs, it can be twice worse than $\ProcMain$. Let $q \geq 0$ be any invoking cost and $K_{1,n}$ be a star rooted in the internal vertex with edges of the weight $l > q$. While $DFS'$ produces the cost of $c' = q + 2ln - l$, the optimal solution is $c = qn + ln$. The ratio $c'/c$ grows to $2$ with the growth of $l$ and $n$.

\section{Trees in the On-line Setting}\label{sec:TreesOn}

In this Section we take a closer look at the algorithms for trees in the on-line setting. Because the height of tree $T$ is not known, the upper bound of the cost, set by $DFS'$, is $q + 2w(T) - \epsilon$, where $\epsilon$ is some small positive constant. This leads to the upper bound of $2$ for the competitive ratio.
We are going to prove that it is impossible to construct an algorithm that achieves better competitive ratio than $2$. We state the problem in the on-line setting formally:

\begin{description}
\item[On-line Tree Problem Statement]$\ $\\[2mm]
Given the invoking cost $q$ and the homebase in the root of $T$, find a strategy of the minimum cost for any tree $T$.
\end{description}

Denote as $\cT \subset \cG$ an infinite class of rooted in $v_0$ trees, where every edge has weight equal to $1$.
For every integer $l \in \natplus$, $i \in \{1,\ldots, l\}$ and $l_i \in \{1,\ldots,l\}$, we add to the class $\cT$ a tree constructed in the following way:

\begin{itemize}
    \item construct $l+1$ paths $P(v_i,v_{i+1}),\ i\in \{0,\ldots,l\}$ of the length $l$;
    \item for every $i \in \{1,\ldots, l\}$ construct a path $P(u'_i,u_i)$ of the length $l_i - 1$ (if $l_i = 1$, then $u_i = u'_i$) and add edge $(v_i,u'_i)$.
\end{itemize}

In other words, every graph in $\cT$ has a set of decision vertices $\{v_1,\ldots,v_l\}$ and set of leaves $\{v_{l+1}, u_1,\ldots,u_{l}\}$. Every decision vertex has exactly two children, $v_i$ is an ancestor of $v_j$ and $d(v_i,v_{i+1}) = l$ for every $0 \leq i < j \leq l + 1$.  
See Figure \ref{fig:onlineA}.

\begin{figure}[htb]
  \centering
\includegraphics[scale=1]{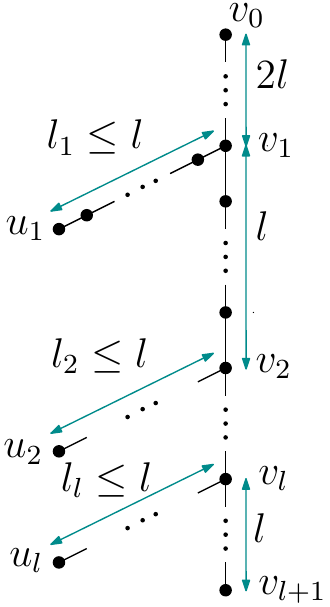}
\caption{
Illustration of graphs from the class $\cT$, where $l \in \natplus$ and $l_i \in \{1,\ldots,l\}$, $i \in \{1,\ldots, l\}$.}
\label{fig:onlineA}
\end{figure}

\begin{theorem}
Any on-line cost-optimal solution for trees is $2$-competitive. 
\end{theorem}
\begin{proof}
$DFS'$ is an example of an algorithm at most twice worse than the best solution, which sets the upper bound.
We are going now to show that for any invoking cost $q \geq 0$ and strategy $\cS$ there exists a tree $T \in \cT$ and a strategy $\cS'$, such that $\cS(T) \geq 2\cS'(T)$. Let $l \in \natplus$ be any integer. Values of $l_i,\ i \in \{1,\ldots,l\}$ are set during the execution of $\cS$. For every $v_i,\ i \in \{1,\ldots,l\}$ three cases can occur.

\paragraph*{Case A: \textmd{\textit{More than one agent reaches $v_i$ before any child of $v_i$ is explored.}}} The value of $l_i$ is set as $1$.

\paragraph*{Case B: \textmd{\textit{One of the agents explores one of the branches of $v_i$ at the depth $0 \leq h < l$ and the second branch at the depth $l$, before any other agent reaches $v_i$ for the first time.}}} In this situation we choose a set of graphs from $\cT$ for which the explored vertex at the depth $l$ is $v_{i+1}$ and $l_i = h + 1$. The value of $h$ might be $0$, as it takes place e.g., for $DFS$.

\paragraph*{Case C: \textmd{\textit{One of the agents explores two branches of $v_i$ at the depth $0 \leq h_1 < l,\ 1 \leq h_2 < l$, before any other agent reaches $v_i$ for the first time.}}} Without loss of generality, we assume that the branch explored to the level $h_2$ is visited as the last one. In this situation we choose a set of graphs from $\cT$ for which vertex $v_{i+1}$ belongs to the branch of $v_i$ explored till the level $h_2$ and $l_i = h_1 + 1$. Once again, $h_1 = 0$ means that the branch was not explored at all.

When $\cS$ explores $v_{l+1}$, all $l_i$ are defined and the set of graphs is narrowed to the exactly one graph, which we denote as $T$. 
We claim first that the distance $d_0$ traversed along the path $P(v_1,v_{l+1})$ is at least $2l^2 - l$ in any $\cS$. Let agent $a_1$ be the one, which explores $v_{l+1}$ and 
let $k \in \{0,\ldots, l\}$ be the number of decision vertices visited by more than one agent. 

\paragraph*{Case A': \textmd{\textit{$k = 0$, i.e., $T$ is explored by one agent.}}} In other words, for all $v_i$ holds the Case B. We notice that, whenever a strategy $\cS$ explores $v_{i+1},\ i \in \{1,\ldots,l\}$, exactly one vertex (i.e., leaf $u_i$) on the path $P(v_{i},u_{i})$ is unexplored. Thus, $P(v_1,v_{l+1})$ has to be traversed at least twice and $d_0 \geq 2l^2$.

\paragraph*{Case B': \textmd{\textit{$k = l$.}}} Path $P(v_1,v_{l})$ has to be obviously traversed at least twice and $P(v_l,v_{l+1})$ once, i.e., $d_0 \geq 2l(l-1) + l= 2l^2 - l$.

\noindent
\paragraph*{Case C': \textmd{\textit{$0 < k < l$.}}} In other words, $T_{v_{k+1}}$ is explored by one agent. Paths $P(v_1,v_k)$ and $P(v_{k+1}, v_{l+1})$ are traversed at least twice and $P(v_k,v_{k+1})$ at least once. Thus, $d_0 \geq 2l(k-1) + 2l(l-k) + l= 2l^2 - l$. 

Now, we have to analyze paths $P(v_i,u_i),\ i\in \{1,\ldots,l\}$. We divide decision vertices into the four groups based on the performance of $\cS$:
\begin{itemize}
\item $V_1 = \{v_i | l_i = 1,\ \text{no agent terminates in }u_i,\ i \in \{1,\ldots,l\}  \}$;
\item $V_2 = \{v_i | l_i = 1,\ \text{at least one agent terminates in }u_i,\ i \in \{1,\ldots,l\} \}$;
\item  $V_3 = \{v_i | l_i > 1,\ \text{no agent terminates in any vertices of the path}$ \\ $P(v_i,u_i),\ i \in \{1,\ldots,l\} \}$;
\item $V_4 = \{v_i | l_i > 1,\ \text{at least one agent terminates in a vertex from the path}$ \\ $ P(v_i,u_i),\ i \in  \{1,\ldots,l\} \}$.
\end{itemize}

Notice that $V_1, V_2, V_3$ and $V_4$ form a partition of decision vertices. Let us denote as $d_i$ the total distance traversed by all the agents along $P(v_i,u_i)$ in $\cS$. For any $v_i \in V_1$ we have $d_i \geq 2$ and $v_i \in V_2$ we have $d_i \geq 1$.
From the way how $T$ is constructed follows, that if $l_i>1$, then either holds Case B and $h > 0$ or Case C and $h_1 > 0$. In both situations path $P(v_i, p(u_i))$ is first traversed at least twice, leaving $u_i$ unexplored. If now, no agent terminates in any vertex of $P(v_i,u_i)$, then $P(v_i,p(u_i))$ has to be traversed at least twice more. Thus, 
\begin{equation}
    d_i \geq 4(l_i - 1) + 2 \geq 4l_i - 2, \qquad v_i \in V_3.
\end{equation}

On the other hand, if at least one agent terminates in any vertex of $P(v_i,u_i)$, then $P(v_i,p(u_i))$ can be traversed only one extra time. Which leads to, 
\begin{equation}
    d_i \geq 3(l_i - 1) + 1 \geq 3l_i - 2, \qquad v_i \in V_4.
\end{equation}

Lastly, we have to consider the extra cost $d'$ generated by searchers. Every invoked agent, which terminates on some path $P(v_i,u_i)$ has to traverse the edge $(v_0,v_1)$, thus 
\begin{equation}
d' \geq (q + l)\left(|V_2| + |V_4|\right) \geq |V_2| + l|V_4|.
\end{equation}

The total cost of exploring $T$ by $\cS$ can be lower bounded by
\begin{align}
&\cS(T) \geq 2l^2 - l + 2|V_1| + |V_2| + \sum\limits_{v_i \in V_3}(4l_i-2) + \sum\limits_{v_i \in V_4}(3l_i-2) + |V_2| + \\
& + l|V_4| \geq 2l^2 - l + 2|V_1| + 2|V_2| + 4\sum\limits_{v_i \in V_3}l_i - 2|V_3| + 4\sum\limits_{v_i \in V_4}l_i -  2|V_4|  = \\
& = 2l^2 - l + 4\sum\limits_{i=1}^l l_i - 2(|V_1| + |V_2| + |V_3| + |V_4|) =  2l^2 + 4\sum\limits_{i=1}^l l_i -3l.
\end{align}

Consider now the following off-line strategy $\cS'$, which explores the same graph $T$ by using one agent, which after reaching the decision vertex $v_i,\ i \in \{1,\ldots,l\}$, firstly traverses the path $P(v_i,u_i)$, then returns to $v_i$ and explores further the tree. The agent finally terminates in $v_{l+1}$. Thus, the path $P(v_0,v_{l+1})$ of the length $(l+1)l$ is traversed only once and paths $P(v_i, u_i),\ i \in \{1,\ldots,l\}$ twice. The optimal strategy can be then upper bounded by
\begin{equation}
\cS^{opt}(T) \leq \cS'(T) = q + l^2 + 2\sum\limits_{i=1}^{l}l_i + l.
\end{equation}

This leads to the following competitive ratio 
\begin{align}
r(\cS) &= \lim\limits_{l \rightarrow \infty} \frac{\cS(T)}{\cS^{opt}(T)} \geq \lim\limits_{l \rightarrow \infty} \frac{2l^2 + 4\sum\limits_{i=1}^{l}l_i - 3l}{q + l^2 + 2\sum\limits_{i=1}^{l}l_i + l}=\\ &= 2 - \lim\limits_{l \rightarrow \infty} \frac{5l +2q}{q + l^2 + 2\sum\limits_{i=1}^{l}l_i + l} 
=2,
\end{align}

which finishes the proof.
\end{proof}

\section{Conclusion}\label{sec:conclusion}

In this work we propose a new cost of the team exploration, which is the sum of total traversed distances by agents and the invoking cost which has to be paid for every searcher. This model describes well the real life problems, where every traveled unit costs (e.g., used fuel or energy) and entities costs itself (e.g., equipping new machines or software license cost).  
The algorithms, which construct the cost-optimal strategies for any given edge-weighted ring and tree in $O(n)$ time are presented. As for the on-line setting the $2$-competitive algorithm for rings is given and the lower bounds of $3/2$ and $2$ for the competitive ratio for rings and trees, respectively, are proved. While there is very little done in this area, a lot of new questions have been pondered. Firstly, it would be interesting to consider other classes of graphs, also for the edge-exploration (where not only every vertex has to be visited, by also every edge). Intuitively, for some of them, the problem would be easy and for some might be NP-hard (e.g., cliques). Another direction is to look more into the problem in the on-line setting, which is currently rapidly expanding due to its various application in many areas. It would be highly interesting to close the gap between lower and upper bounds of the competitive ratio for rings. Another idea is to bound communication for agents, which will make this model truly distributed. One may notice, that a simple solution of choosing one leader agent to pass messages between the other entities might not be cost-optimal, as it significantly rises the traveling cost.
Lastly, different variation of this model might be proposed, e.g., the invoking cost might increase/decrease with the number of agents in use or time might be taken under consideration as the third minimization parameter.

\bibliographystyle{plain}
\bibliography{bib}

\end{document}